\documentclass[letterpaper, 10 pt, conference]{ieeeconf}  

\IEEEoverridecommandlockouts                              

\usepackage{graphicx}
\let\labelindent\relax
\usepackage{sty/nlab-settings}
\usepackage{sty/nlab-commands}

\usepackage{sty/nlab-theorems}
\usepackage{sty/mymacros}
\usepackage{mathtools}
\usepackage{dsfont}

\title{\LARGE \bf
Periodic Event-Triggered Explicit Reference Governor for Constrained Attitude Control on $SO(3)$
}

\author{Satoshi Nakano$^{1}$, Masahiro Suzuki$^{1}$, Misa Ohashi$^{1}$, Noboru Chikami$^{1}$, and Shusuke Otabe$^{1}$
\thanks{*This work was supported by JSPS KAKENHI Grant Numbers JP21K03308, JP21K20334, JP23K13352, JP24K16894}
\thanks{$^{1}$The authors are with the Department of Engineering, Nagoya Institute of Technology, Nagoya, 4668555 Aichi, Japan
        {\tt\small \{nakano, masahiro, ohashi.misa, chikami.noboru, otabe.shusuke\}@nitech.ac.jp}}%
}

\begin{document}

\maketitle
\thispagestyle{empty}
\pagestyle{empty}

\begin{abstract}
This letter addresses the constrained attitude control problem for rigid bodies directly on the special orthogonal group $SO(3)$, avoiding singularities associated with parameterizations such as Euler angles.
We propose a novel Periodic Event-Triggered Explicit Reference Governor (PET-ERG) that enforces input saturation and geometric pointing constraints without relying on online optimization.
A key feature is a periodic event-triggered supervisory update: the auxiliary reference is updated only at sampled instants when a robust safety condition is met, thereby avoiding continuous-time reference updates and enabling a rigorous stability analysis of the cascade system on the manifold.
Through this structured approach, we rigorously establish the asymptotic stability and exponential convergence of the closed-loop system for almost all initial configurations.
Numerical simulations validate the effectiveness of the proposed control architecture and demonstrate constraint satisfaction and convergence properties.
\end{abstract}

\section{INTRODUCTION}
\label{sec:introduction}
Many applications require rigid bodies to perform precise reorientation maneuvers, and attitudes are naturally represented by rotation matrices on the special orthogonal group $SO(3)$ \cite{lee_11acc, lee_13tcst, chaumette_06ram}. Common parameterizations (Euler angles, quaternions) introduce singularities or topological complications, motivating controller design directly on $SO(3)$ \cite{chaturvedi_11csm}.

Reorientation tasks frequently face strict, simultaneous constraints: geometric pointing constraints (e.g., avoiding the Sun) and actuator saturation. While Model Predictive Control (MPC) handles constraints explicitly, it demands online optimization; artificial potential or barrier-based methods are computationally light but often lack rigorous treatment of input saturation \cite{lee_11acc, kalabic_17automatica, kulumani_17ijc}. Reference Governor (RG) methods provide a complementary, optimization-free approach by modulating the reference to a pre-stabilized system, and the Explicit Reference Governor (ERG) in particular uses invariant sets and offline safety margins to enforce constraints without online optimization \cite{garone_16tac, nicotra_19tac}. Recent event-triggered RG variants aim to reduce updates or communication load \cite{an2024}, but are predominantly developed in Euclidean settings.

In our prior work \cite{nakano2023}, we introduced a continuous-time ERG on $SO(3)$ to handle these specific constraints.
Nevertheless, a continuous-time formulation poses a fundamental theoretical challenge: rigorously proving the overall asymptotic stability of the interconnected closed-loop system, i.e., formally guaranteeing that the actual attitude $R$ converges to the desired attitude $R_d$, rather than just the auxiliary reference $R_g$ converging to $R_d$, is highly non-trivial.
Because the continuous-time reference update is tightly coupled with the nonlinear state-dependent Lyapunov function, the update rate can become highly sensitive or even vanish near the boundary of the safe set, complicating the analysis of the full cascade convergence.
This tight coupling makes it formidably complex to analytically guarantee that the total variation of the reference signal remains bounded, which is a critical prerequisite for establishing the exact convergence of the full nonlinear cascade system.

To resolve this analytical bottleneck without sacrificing the smooth tracking performance of continuous-time feedback, we propose a Periodic Event-Triggered ERG (PET-ERG) that evaluates a robust safety condition exclusively at discrete sampling instants. 
Depending on this evaluation, the auxiliary reference is either continuously updated using continuous state measurements or held strictly constant until the next sampling instant. 
While conceptually related to Periodic Event-Triggered Control (PETC) frameworks \cite{heemels2013}, our use of periodic event-triggering is not intended for communication or energy reduction, unlike conventional PETC. 
Instead, it is deliberately introduced as a supervisory mechanism to enforce a finite minimum duration of reference updates. 
By effectively decoupling the reference kinematics from the microscopic fluctuations of the inner-loop dynamics, this finite duration serves as the key property that enables bounding the total variation of the reference signal and, in turn, a rigorous stability analysis of the cascade system on $SO(3)$. 
Consequently, this structured hybrid architecture provides the necessary analytical properties to bridge the gap between $\lim_{t \to \infty}\|R - R_g\|=0$ and $\lim_{t \to \infty}\|R_g - R_d\|=0$, allowing us to rigorously establish the asymptotic stability and exponential convergence of the overall closed-loop system for almost all initial configurations under simultaneous geometric and input constraints, and validate the approach via numerical simulations.

\subsection*{Notation and Preliminaries:}
The special orthogonal group is defined as $SO(3) = \{\ro \in \R^{3 \times 3} \mid \ro^\T \ro = \ma{I}_3, \det \ro = 1\}$, and its Lie algebra is $\mathfrak{so}(3) = \{\ma{S} \in \R^{3 \times 3} \mid \ma{S}^\T = -\ma{S}\}$.
The tangent space to $SO(3)$ at $\ro \in SO(3)$ is denoted by $T_{\ro} SO(3) = \{\ma{R} \ma{S} \mid \ma{S} \in \mathfrak{so}(3) \}$.
The notation $\|\cdot\|$ denotes the Euclidean norm for vectors and the Frobenius norm for matrices.
The unit sphere in $\R^3$ is defined as $\uS^2 = \{\ve{x} \in \R^3 \mid \|\ve{x}\| = 1\}$.

The hat map $\wedge : \R^3 \to \mathfrak{so}(3)$ is defined such that $\hat{\ve{x}} \ve{y} = \ve{x}^\wedge \ve{y} = \ve{x} \times \ve{y}$ for any $\ve{x}, \ve{y} \in \R^3$, where the notations $\hat{\cdot}$ and $(\cdot)^\wedge$ are used interchangeably.
The inverse map is denoted by $(\cdot)^\vee : \mathfrak{so}(3) \to \R^3$ such that $(\ve{x}^\wedge)^\vee = \hat{\ve{x}}^\vee = \ve{x}$.
The skew-symmetric projection operator is defined as $\sk(\ma{A}) = \frac{1}{2} (\ma{A} - \ma{A}^\T)$.
The indicator function $\mathds{1}_{I}$ takes the value $1$ if the argument belongs to the set $I$ and $0$ otherwise.

Finally, we recall Barbalat's Lemma from \cite{slotine_book}: if a differentiable function $f(t)$ has a finite limit as $t \to \infty$ and if $\dot{f}$ is uniformly continuous, then $\dot{f}(t) \to 0$ as $t \to \infty$.

\section{PROBLEM STATEMENT}
\label{sec:description}
This letter considers the constrained attitude control problem for a fully actuated rigid body.
Let $\Sigma_w$ and $\Sigma_b$ denote the inertial and body-fixed frames, respectively.
The attitude of the rigid body is represented by a rotation matrix $\rmr{} \in SO(3)$, which maps coordinates from $\Sigma_b$ to $\Sigma_w$.
The attitude kinematics and dynamics are governed by the following equations \cite{murray_book}:
\begin{align}
\label{eq:rot-sys}
    \rmr[d]{} = \rmr{} \hat{\ve{\omega}}, \quad
    \ma{J} \dot{\ve{\omega}} & = (\ma{J} \ve{\omega})^\wedge \ve{\omega} + \ve{\tau},
\end{align}
where $\ve{\omega} \in \mathbb{R}^3$ is the angular velocity vector expressed in $\Sigma_b$, $\ma{J} \in \mathbb{R}^{3 \times 3}$ is the symmetric positive-definite inertia matrix, and $\ve{\tau} \in \mathbb{R}^3$ is the control torque input.

To ensure safe operation and respect actuator limits, the system is subject to the following pointwise-in-time constraints:
\begin{subequations}
\label{eq:constraint}
\begin{align}
    \label{eq:constraint_input}
    \|\ve{\tau}(t)\| &\leq \tau_{\max}, \quad \forall t \ge 0, \\
    \label{eq:constraint_pointing}
    \ve{a}_{c}^\T \rmr{}(t) \ve{a}_{b} &\geq \cos \theta_{c}, \quad \forall t \ge 0,
\end{align}
\end{subequations}
where $\tau_{\max} > 0$ is the maximum available torque. The geometric pointing constraint \eqref{eq:constraint_pointing} ensures that a sensitive body-fixed sensor, aligned with the axis $\ve{a}_b \in \uS^2$, maintains a minimum angular distance $\theta_c$ from a critical inertial direction $\ve{a}_c \in \uS^2$.

Let $\rmr{d}(t) \in SO(3)$ be a (possibly time-varying) desired attitude trajectory.
The control objective is to design a control law that guarantees: (i) \textbf{Asymptotic Tracking}: for any constant steady-state admissible desired attitude $\rmr{d}$, $\lim_{t \to \infty} \rmr{}(t) = \rmr{d}$; and (ii) \textbf{Constraint Satisfaction}: the specific operational constraints \eqref{eq:constraint} are satisfied at all times for any continuous reference $\rmr{d}(t)$.

To address this problem, we employ an Explicit Reference Governor (ERG) framework.
The core idea is to manipulate an auxiliary reference $\rmr{g}$ based on the invariant level sets of a Lyapunov function $V$.
Specifically, we derive a scalar function $\Gamma(\rmr{g})$ that defines a safe sublevel set of $V$, ensuring that constraints are never violated as long as the state remains within this set.
This invariance property forms the basis of the periodic event-triggered ERG strategy detailed in Section \ref{sec:erg_overview}.

\section{CONTROL ARCHITECTURE}
\label{sec:scheme}

To address the constrained attitude control problem, we employ a cascade control architecture depicted in Fig.~\ref{fig:block}. The scheme comprises two main subsystems: an \textit{Attitude Stabilization Controller (Inner Loop)} that stabilizes the rigid body attitude $\rmr{}$ to an auxiliary reference $\rmr{g}$, and an \textit{Explicit Reference Governor (Outer Loop)} that manipulates $\rmr{g}$ to ensure constraint satisfaction while guiding the system toward the desired reference $\rmr{d}$.
This letter builds upon the ERG framework specifically designed for $SO(3)$ in our previous work \cite{nakano2023}. By exploiting the invariant sets of the pre-stabilized inner loop, the ERG manages constraints without the need for online optimization.

\begin{figure}[t]
    \centering
    \includegraphics[width=\linewidth]{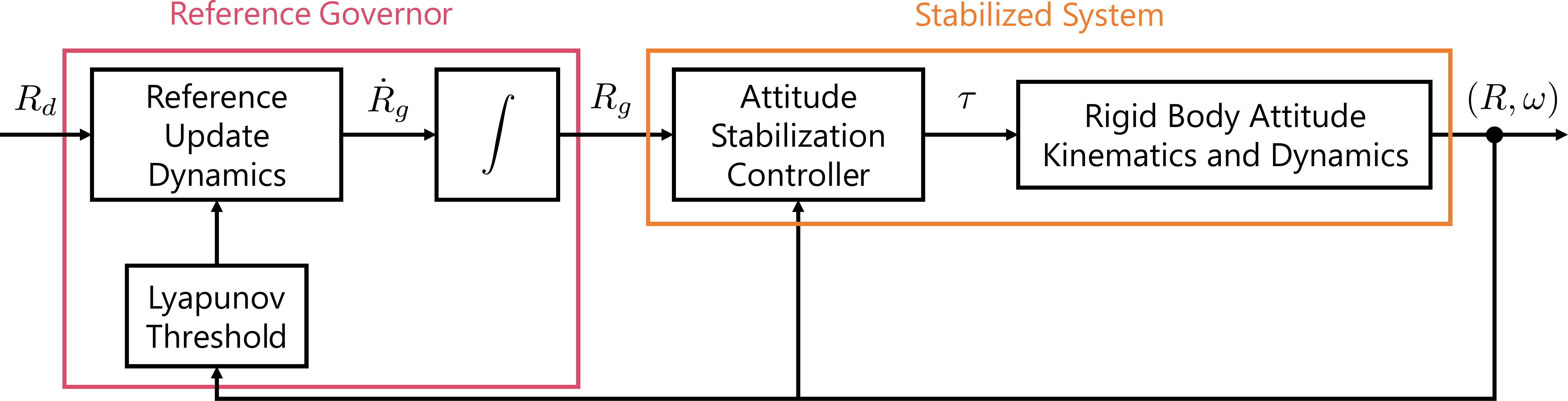}
    \caption{Block diagram of control scheme.}
    \label{fig:block}
\end{figure}

\subsection{Attitude Stabilization Controller}
\label{sec:prestabilization}
We first design a nominal control law to stabilize the attitude dynamics \eqref{eq:rot-sys} to a constant reference $\rmr{g}$. At this stage, constraints \eqref{eq:constraint} are disregarded, as they will be handled by the ERG in the outer loop.

Consider the proportional-derivative (PD) control law on $SO(3)$ \cite{bullo_book}:
\begin{align}\label{eq:taulaw}
    \ve{\tau}(\rmr{}, \ve{\omega}, \rmr{g}) = \gain{p} \sk (\rmr{e})^\vee - \gain{d} \ve{\omega},
\end{align}
where $\gain{p}, \gain{d} > 0$ are tuning gains and $\rmr{e} := \rmr[t]{} \rmr{g} \in SO(3)$ represents the attitude error.
For any constant reference $\rmr{g}$, the asymptotic stability of the closed-loop system formed by \eqref{eq:rot-sys} and \eqref{eq:taulaw} is established using the Lyapunov function candidate:
\begin{align*}
    V(\rmr{}, \ve{\omega}, \rmr{g}) & = \frac{1}{2} \ve{\omega}^\T \ma{J} \ve{\omega} + \gain{p} \phi(\rmr{e}), \\
    \phi (\rmr{e})                  & = \frac{1}{2} \tr(\ma{I}_3 - \rmr{e}) \geq 0.
\end{align*}
The time derivative of $V$ along the trajectories of the closed-loop system satisfies $\dot{V} = -\gain{d} \|\ve{\omega}\|^2 \leq 0$.
By invoking LaSalle's invariance principle, it follows that $(\rmr{}, \ve{\omega}) \to (\rmr{g}, \ve{0})$ asymptotically for any constant $\rmr{g}$.

\subsection{Explicit Reference Governor (ERG)}
\label{sec:erg_overview}

The control law presented in Section~\ref{sec:prestabilization} stabilizes the system to a constant reference $\rmr{g}$ but does not account for constraints.
To enforce constraints during transients, we augment the closed-loop system with an ERG.
The ERG is a supervisory scheme that continuously manipulates the auxiliary reference $\rmr{g}$ according to the differential equation $\rmr[d]{g} = \Delta(V,\rmr{g}) \, \ma{\rho}_n(\rmr{d},\rmr{g})$, where $\Delta(V,\rmr{g}) \in \R_{\ge 0}$ is the \emph{Dynamic Safety Margin} (DSM) and $\ma{\rho}_n(\rmr{d},\rmr{g}) \in T_{\rmr{g}}SO(3)$ is the \emph{Navigation Field} (NF).
In this work, we adopt the specific DSM and NF formulations proposed in \cite{nakano2023} to ensure rigorous stability properties on the manifold.

\subsubsection{Dynamic Safety Margin (DSM)}
The DSM guarantees that the system state strictly satisfies the specific operational constraints \eqref{eq:constraint}.
By leveraging the forward invariance of the sublevel sets of $V(\rmr{}, \ve{\omega}, \rmr{g})$ for a fixed $\rmr{g}$, we construct a state-dependent safety bound $\Gamma(\rmr{g})$. This bound ensures that the invariant level set $\mathcal{I}(\Gamma) := \{ (\rmr{},\ve{\omega}) \mid V(\rmr{},\ve{\omega},\rmr{g}) \le \Gamma(\rmr{g}) \}$ is contained in the interior of the admissible set $\mathcal{C} = \mathcal{C}_d \cap \mathcal{C}_g$, where $\mathcal{C}_d$ and $\mathcal{C}_g$ represent the states satisfying the input constraint \eqref{eq:constraint_input} and the pointing constraint \eqref{eq:constraint_pointing}, respectively.

The aggregate threshold is defined as $\Gamma(\rmr{g}) = \min \{\Gamma_d, \Gamma_g (\rmr{g}), k_p(2 - \varepsilon_\Gamma)\}$, where $\Gamma_g(\rmr{g})$ handles $\mathcal{C}_g$ and is derived from the closed-form geometry detailed in \cite{nakano2023}, and $\varepsilon_\Gamma \in (0, 2)$ is a small constant.
Regarding $\Gamma_d$, we reduce the conservatism of our prior work \cite{nakano2023} by solving the associated offline optimization problem directly on $SO(3) \times \R^3$ to ensure the forward invariance of $\mathcal{C}_d$.
Finally, the DSM is constructed as $\Delta(V,\rmr{g}) = \kappa \max \{ \Gamma(\rmr{g}) - V(\rmr{},\ve{\omega},\rmr{g}), 0 \}$ with $\kappa > 0$. This ensures that $\dot{\rmr{g}} = 0$ whenever $V \ge \Gamma(\rmr{g})$, strictly enforcing constraint satisfaction.

\subsubsection{Navigation Field (NF)}
The Navigation Field determines the direction of motion for $\rmr{g}$ toward $\rmr{d}$.
We employ the specific NF proposed in \cite{nakano2023}, defined as the negative gradient of an artificial potential function $P(\rmr{g}, \rmr{d})$ with respect to $\rmr{g}$: $\ma{\rho}_n(\rmr{d},\rmr{g}) = - \gr P(\rmr{g}, \rmr{d})$.
The potential function $P = P_a + P_r$ consists of an attractive term minimized at $\rmr{g} = \rmr{d}$ and a repulsive term that diverges near the boundary of the steady-state admissible set. This repulsive property safely guides $\rmr{g}$ around the non-convex forbidden regions induced by \eqref{eq:constraint_pointing}. Crucially, as established in \cite{nakano2023}, this NF guarantees almost global convergence to $\rmr{d}$ (except for a set of measure zero corresponding to the undesired equilibria induced by the topology of $SO(3)$) despite the non-convexity of the constraints. This property is essential for the periodic event-triggered stability analysis presented in Section~\ref{sec:proposed_method}.

\section{PERIODIC EVENT-TRIGGERED REFERENCE GOVERNOR ON $SO(3)$}
\label{sec:proposed_method}

\subsection{Proposed Update Mechanism}
While the continuous-time ERG on $SO(3)$ guarantees convergence of $\rmr{g}$ to $\rmr{d}$ and safety of the closed-loop system under the assumption of a constant reference or a static auxiliary reference, establishing the stability of the interconnected closed-loop system (cascade of the pre-stabilized dynamics and the ERG) is non-trivial.
To address this, we propose a \emph{periodic event-triggered} implementation of the reference update law.
This approach allows for a rigorous stability analysis of the overall system.

In the continuous formulation, the auxiliary reference dynamics are given by:
\begin{align}
\dot{R}_g(t) = \kappa \left(\Gamma(\rmr{g}) - V\right) \frac{\rho_n(\rmr{d}, \rmr{g})}{\max\{\|\rho_n(\rmr{d}, \rmr{g})\|,\varepsilon\}}.
\label{eq:cont_update}
\end{align}
In this work, we reformulate~\eqref{eq:cont_update} within the framework of \emph{periodic event-triggered control} (PETC).

Let $\mathcal{T} = \{t_k\}_{k\in\mathbb{N}_0}$ denote the set of sampling instants, defined as $t_{k+1} = t_k + T_s$ with a fixed sampling period $T_s>0$.
At each sampling instant $t_k \in \mathcal{T}$, the safety condition is evaluated.
The reference $\rmr{g}$ is updated during the interval $[t_k, t_{k+1})$ only if the following condition holds:
\begin{align}
\Gamma(\rmr{g}(t_k)) - c_\Gamma V(t_k) \ge 0.
\label{eq:event_condition}
\end{align}
where $c_\Gamma > 1$ is a tuning parameter.
This condition acts as a robust safety margin, ensuring that the Lyapunov function $V$ remains bounded by $\Gamma$ even during the inter-sample behavior.
Choose $T_s$ and $c_\Gamma>1$ so that $s_{\#}:=\min\{T_s, (c_\Gamma-1)c_1/(2 c_\Gamma C)\}>0$ (with constants defined later in Theorem 1); this ensures the positive margin over each triggered interval.
The PETC-based update law is defined as:
\begin{align}
\label{eq:ref_update}
    \dot{R}_g(t) &=
    \kappa \big(\Gamma(t) - V(t)\big) \mathds{1}_{\{(\Gamma-c_\Gamma V)(\sigma(t))\geq 0\}} \notag \\
    &\quad \times \frac{\rho_n(t)}{\max \{\|\rho_n(t)\|, \varepsilon\}}, \quad \forall t \in [t_k, t_{k+1}),
\end{align}
where $\sigma(\tau) := \max \{t_k \in \mathcal{T} \mid t_k \le \tau\}$.
The update law \eqref{eq:ref_update} implements a sample-and-hold strategy: if the safety condition is met at $t_k$, the reference moves according to the Navigation Field evaluated at $t_k$; otherwise, it remains constant.
The overall closed-loop system consists of the attitude dynamics \eqref{eq:rot-sys} with control law \eqref{eq:taulaw} and the reference update law \eqref{eq:ref_update}, with initial conditions $R(t_0)=R_{0}$, $\omega(t_0)=\omega_{0}$, and $R_{g}(t_0)=R_{0}$.

\subsection{Stability Analysis}
We now analyze the stability of the proposed closed-loop system.
The time derivative of the Lyapunov function $V$ along the trajectories of the closed-loop system satisfies:
\begin{align}\label{eq:V}
	\dot{V} = -\frac{k_{p}}{2} \tr \big(R^\T \dot{R}_{g}\big) - k_{d} \|\omega\|^{2}.
\end{align}
Note that if $\dot{R}_{g} \equiv 0$, $V$ serves as a standard Lyapunov function for the autonomous attitude dynamics.

Let $D \subset SO(3)$ be the domain of attraction for the auxiliary reference dynamics: $D := \{ R_0 \in SO(3) \mid \exists R_g(t) \text{ s.t. } R_g(t_0)=R_0, \dot R_g = \rho_n(R_d,R_g), \lim_{t\to\infty} R_g(t)=R_d \}$.
Due to the properties of the designed Navigation Field, this set covers almost all of $SO(3)$ except for a set of measure zero. This excludes only the undesired equilibria induced by the topology of $SO(3)$ (e.g., antipodal configurations), as established in \cite{nakano2023}.
Throughout this letter, we assume the initial auxiliary reference and plant state satisfy $(R_0,\omega_0,R_{g}(t_0))\in D\times\mathbb{R}^3 \times D$; in particular, $R_{g}(t_0)$ is chosen within the domain of attraction of the navigation field.

In the following proofs, $C$ denotes a generic positive constant whose exact value may change from line to line to simplify the notation.
We emphasize that $R_g$ is continuous and continuously differentiable on each open sampling interval; all derivative-based inequalities are applied on such intervals where the vector field is Lipschitz. By enforcing the event condition only at sampling instants and proving the infinite occurrence of safe samples (cf. Eq. \eqref{Sinfity1}), we obtain a monotone increasing time-rescaling $S(t)$ whose growth guarantees that the total action of $\dot{R}_g$ is bounded in the required sense; details appear in the proof of Theorem \ref{Thm1}.

We are now in a position to state the main theorem regarding the global solvability and convergence of the proposed system.
\begin{thm}
	\label{Thm1}
	Consider the closed-loop system with initial conditions $(\omega_{0},R_{0}) \in \mathbb{R}^{3} \times D$ satisfying $V(R_{0},\omega_{0},R_{0}) \leq \Gamma(R_{0})$.
	Then the system admits a unique global solution $(R,\omega,R_{g})$ satisfying $R, \omega \in C^{1}([t_0,\infty))$ and $R_{g} \in C([t_0,\infty)) \cap C^{1}\big([t_0,\infty)\setminus \{t_k\}_{k\in\mathbb{N}_0}\big)$.
    Furthermore, the following properties hold:
    (i) $V(R,\omega,R_{g})(t) \leq \Gamma(R_{g})(t)$ for all $t \ge t_0$;
    (ii) $\lim_{t\to\infty} R_{g}(t) = R_{d}$; and
    (iii) $\lim_{t \to \infty} R(t) = R_{d}$ and $\lim_{t \to \infty} \omega(t) = 0$.
\end{thm}
\begin{proof}
	The time-local solvability follows from standard ODE theory, since the right-hand sides of the closed-loop dynamics are locally Lipschitz continuous on each interval $[t_k,t_{k+1})$.
	To establish time-global solvability, we assume that the closed-loop system admits a solution $(R,\omega,R_{g})$ on $[t_0, T)$ for some $T>t_0$, and derive the \textit{a priori} estimate:
	\begin{align}\label{aes1}
		\sup_{t \in [t_0,T]}( \|R(t)\|+\|\omega(t)\|+\|R_{g}(t)\|)
        \leq  C(1+\|\omega_{0}\|)e^{CT},
	\end{align}
	where $C>0$ is a constant independent of $T$.
	Standard arguments from the theory of differential equations guarantee that such an \textit{a priori} estimate, combined with time-local solvability, ensures the existence of a unique global solution.
	
	Let us prove \eqref{aes1}.
	Since $R, R_{g} \in SO(3)$, their norms are uniformly bounded by a constant $C>0$.
	It remains to estimate $\omega(t)$.
    Using \eqref{eq:ref_update} and \eqref{eq:V}, and noting that $\Gamma(\rmr{g})$ is bounded and $\|\frac{\rho_{n}}{\max\{\|\rho_{n}\|,\varepsilon\}}\| \le 1$, we observe that $\dot{V} \leq -\frac{k_{p}}{2} \tr(R^\T \dot{R}_{g}) \leq C\|\dot{R}_{g}\| \leq C(1+V)$.
	Applying a standard differential inequality argument yields $\|\omega\|^{2} \leq C(1+\|\omega_{0}\|^{2})e^{CT}$.
	This establishes the time-global solvability of the closed-loop system.

	We next show property (i).
    A similar result was established in \cite{nakano2023}, but the proof here is reconstructed to explicitly account for the periodic update mechanism.
	To this end, we consider the system where the reference update law \eqref{eq:ref_update} is replaced by
	\begin{align}\label{req3}
		\dot{R}_{g}&=  \max\{0,\Gamma-V\}\mathds{1}_{\{(\Gamma-c_\Gamma V)(\sigma(t))\geq 0\}} \frac{\kappa \rho_n(t)}{\max \{\|\rho_n(t)\|, \varepsilon\}}, \notag \\
    &\qquad \forall t \in [t_k, t_{k+1}).
	\end{align}
	for all $k \in \mathbb{N}_0$.
	Analogous reasoning establishes the time-global solvability of this modified system.
	We claim that the solution satisfies property (i).
	This is true at $t=t_0$ by assumption.
	Proceeding by contradiction, assume that $V(\bar{t}) > \Gamma(\bar{t})$ for some $\bar{t}>t_0$.
	Define $\underline{t}:=\sup\{t \in (t_0,\bar{t}) \mid V(t) \leq \Gamma(t) \}$.
	Note that $\underline{t}<\bar{t}$, $V(\underline{t}) = \Gamma(\underline{t})$, and $V(t) - \Gamma(t) > 0$ for all $t \in (\underline{t},\bar{t}]$.
	Consequently, $\dot{R}_{g}=0$ on $(\underline{t},\bar{t}]$ due to \eqref{req3}.
    Since $\Gamma=\Gamma(R_g)$, we have $\dot{\Gamma}=0$ on $(\underline{t},\bar{t}]$ (as $\dot{R}_g=0$), and from \eqref{eq:V}, $\dot{V}=-k_d\|\omega\|^2\le0$.
	Therefore, $\dot{V} - \dot{\Gamma} \le 0$ on $(\underline{t}, \bar{t}]$. Since $V(\underline{t}) = \Gamma(\underline{t})$, this implies that $V(t) - \Gamma(t) \le 0$ for all $t \in (\underline{t}, \bar{t}]$, which contradicts the fact that $V(t) - \Gamma(t) > 0$ on $(\underline{t}, \bar{t}]$.
	Thus, the claim holds, implying $\max\{0,\Gamma-V\}=\Gamma-V$.
	Therefore, the solution of the modified system also solves the original closed-loop system.
	By uniqueness, we conclude that the solution of the closed-loop system satisfies property (i).

    Let us prove property (ii).
	From \eqref{eq:V}, the definition of $\Gamma$, \eqref{eq:ref_update}, and property (i), we have $\dot{\Gamma}-\dot{V} \geq  \dot{\Gamma} -\frac{k_{p}}{2} |\tr(R^\T \dot{R}_{g}) | \geq -C$, where $C$ is a positive constant independent of $t$.
	Furthermore, there exists $c_{1}>0$ independent of $t$ such that
	\begin{align}\label{boundG1}
		\Gamma(R_g(t)) \geq c_{1}>0.
	\end{align}
	Because $R_g(t)$ evolves inside a compact subset strictly contained in the interior of the admissible steady-state set $\mathcal{C}$ for all $t \ge t_0$, and $\Gamma_d, \Gamma_g$ are continuous functions, the extreme value theorem guarantees the existence of such $c_1>0$.
	If $(\Gamma-c_\Gamma V)(t_k) \geq 0$, then for any $t \in [t_k,t_k+ s_{\#}]$,
	\begin{align}
    (\Gamma-V)(t)
     & = (\Gamma-V)(t_k) + \int_{t_k}^{t}  (\dot{\Gamma}-\dot{V}) ds
    \notag \\
     & \geq \frac{c_\Gamma - 1}{c_\Gamma} \Gamma(t_k) + \int_{t_k}^{t}  (\dot{\Gamma}-\dot{V}) ds
    \notag \\
     & \geq  \frac{c_\Gamma - 1}{c_\Gamma} c_{1} \!-\! C(t-t_k) \geq \frac{c_\Gamma - 1}{2 c_\Gamma} c_{1} > 0,
    \label{GammaV1}
\end{align}
	where $s_{\#}:=\min\{T_s, \frac{c_\Gamma - 1}{2 c_\Gamma C} c_1\}$ is independent of $t$.
	We assert that the set of indices $t_k$ such that $(\Gamma-c_\Gamma V)(t_k) \geq 0$ has infinite cardinality.
	Suppose, to the contrary, that $c_\Gamma V(R,\omega,R_{g})(t_k) >  \Gamma(t_k)$ for all $t_k \geq n$ for some $n \in \mathbb N$.
	Then, by \eqref{eq:ref_update}, $\dot{R}_{g}(s)=0$ for all $s \geq t_k \geq n$.
	This implies $\Gamma(R_{g})$ is constant and positive on $[n,\infty)$.
    From \eqref{eq:V}, we deduce that $V(R,\omega,R_{g})(s) \to 0$ as $s \to \infty$.
    However, this contradicts the assumption $c_\Gamma V(s) > \Gamma(s) > 0$.
	Thus, the claim holds.

	From this claim and \eqref{GammaV1}, we have $\lim_{t \to \infty} S(t)=\infty$, where
	\begin{equation}\label{Sinfity1}
		S(t):=\int_{t_0}^{t}(\Gamma(\tau)-V(\tau))\mathds{1}_{\{(\Gamma-c_\Gamma V)(\sigma(\tau))\geq 0\}} d\tau.
	\end{equation}
	On the other hand, it follows from \cite[Lemmas 9--12]{nakano2023} that $\lim_{t \to \infty}\tilde{R}_{g}=R_{d}$, where $\tilde{R}_{g}$ is the solution to the autonomous equation $\frac{d}{ds}{\tilde{R}}_{g}(s)= \frac{\kappa{\rho}_{n}(\tilde{R}_{g}(s))}{{\rm max}\{\|{\rho}_{n}(\tilde{R}_{g}(s))\|,\varepsilon\}}$ with $\tilde{R}_{g}(0)=R_{0} \in D$.
	Define $\hat{R}_{g}(t):=\tilde{R}_{g}(S(t))$. We show that $\hat{R}_{g}\equiv R_{g}$.
    The time derivative of $S(t)$ is given by $\dot{S}(t) = (\Gamma(t)-V(t)) \mathds{1}_{\{(\Gamma-c_\Gamma V)(\sigma(t))\ge 0\}}$.
	Consider the simplified equation $\frac{d{R}_{g}}{dt}(t)= \dot{S}(t) \frac{\kappa\rho_{n}(R_{g}(t))}{{\rm max}\{\|\rho_{n}(R_{g}(t))\|,\varepsilon\}}$.
	It is evident that ${R}_{g}$ satisfies this equation with $R_{g}(0)=R_{0}$.
	Furthermore, $\hat{R}_{g}$ also solves this equation with $\hat{R}_{g}(0)=R_{0}$.
	By the uniqueness of solutions, we conclude $\hat{R}_{g}\equiv R_{g}$.
	Since $\lim_{t \to \infty} S(t)=\infty$, we obtain $\lim_{t\to\infty}R_{g}(t)=\lim_{t\to\infty}\tilde{R}_{g}(S(t))=R_{d}$.
	Thus property (ii) is proved.

	It remains to establish property (iii).
	By property (ii), we know $\lim_{t\to\infty} R_{g}(t) = R_{d}$. Since $\rho_n(R_d, R_d) = 0$, there exists $t_{*}>0$ such that for all $t \geq t_{*}$,
	\begin{subequations}\label{cond1}
		\begin{align}
			P(t)=P_{a}(t),
			\\
			c_{0} \| {\rm grad} P_{a}(t) \|^{2} \leq  P_{a}(R_{g};R_{d}) \leq \frac{1}{c_{0}} \| {\rm grad} P_{a}(t) \|^{2},
			\\
			\|\rho_{n}(R_d, R_g)\| \leq \varepsilon,
		\end{align}
	\end{subequations}
	where $0 < c_{0} < 1$ is a constant independent of $t$.
	Let $\mathcal{T}_{safe}^* := \{ t_k \in \mathcal{T}_{safe} \mid t_k \geq t_{*} \}$.
    As established previously, $\mathcal{T}_{safe}^*$ contains countably infinitely many elements.
    We denote the elements of $\mathcal{T}_{safe}^*$ as a sequence $\{t_{j}\}_{j=1}^{\infty}$ with $t_{*} \leq t_{1}<t_{2}<t_{3}<\ldots$ and $\lim_{j\to\infty} t_{j}=\infty$.
	Note that if $t_k \in \mathcal{T} \setminus \mathcal{T}_{safe}^*$, then $\dot{R}_{g}(t) = 0$ for $t \in (t_k,t_k+T_s)$.
	Using \eqref{eq:ref_update}, \eqref{GammaV1}, and \eqref{cond1}, we observe that
    \begin{align*}
		&\frac{d}{dt} P(R_{g};R_{d})\\
        & =  (\Gamma-V)\mathds{1}_{\{(\Gamma-c_\Gamma V)(\sigma(t))\geq 0\}}  \frac{\kappa\tr({\rm grad}P(R_{g};R_{d})^\nlT \rho_{n}) }{\varepsilon}
		\\
		                            & \leq \left\{
		\begin{array}{ll}
			-\frac{c_{0} c_{1}\kappa (c_\Gamma - 1)}{2 c_\Gamma \varepsilon}  P(R_{g};R_{d}), & t \in [t_k,t_k+s_{\#}],
			\\
			0,                                                       & \text{otherwise},
		\end{array}
		\right.
	\end{align*}
	where  $t_k \in \mathcal{T}_{safe}^*$.
	Then, $p_{j}:=\sup_{t \in [t_{j},t_{j}+T_s]} P(R_{g}(t);R_{d})$ satisfies $p_{j+1}\leq \alpha P(R_{g}(t_{j+1});R_{d}) \leq \alpha p_{j}$,
	where $\alpha:=e^{-\frac{c_{0} c_{1}\kappa (c_\Gamma - 1)}{2 c_\Gamma \varepsilon}s_{\#}} \in (0,1)$.
	Noting that $\| \dot{R}_g(t) \| \le C \| \rho_n(t) \| \le C P(R_g(t); R_d)^{1/2}$ near the equilibrium, we integrate \eqref{eq:V} over $[t_{*},\infty)$ and use \eqref{cond1} and $p_{j+1} \leq \alpha p_{j}$ along with $V\geq 0$ to obtain
	\begin{align*}
		 & k_{d} \int_{t_{*}}^{\infty} \|\omega(\tau)\|^{2} \,d\tau                                                                      \\
		 & \leq  V(t_{*}) - \int_{t_{*}}^{\infty}  \frac{k_{p}}{2} \tr(R(\tau)^\nlT \dot{R}_{g}(\tau)) \,d\tau
		\\
		 & \leq  V(t_{*}) - \sum_{j=1}^{\infty}\int_{t_{j}}^{t_{j}+T_s}  \frac{k_{p}}{2} \tr(R(\tau)^\nlT \dot{R}_{g}(\tau)) \,d\tau
		\\
		 & \leq V(t_{*})  +  C \sum_{j=1}^{\infty} p_{j}^{1/2} <+\infty.
	\end{align*}
    Since the right-hand side of the closed-loop dynamics is piecewise locally Lipschitz and bounded on each interval $[t_k, t_{k+1})$, the time derivative $\frac{d}{dt}\|\omega(t)\|^{2} = 2\omega^\T\dot{\omega}$ is uniformly bounded. This implies $\|\omega(t)\|^{2}$ is uniformly continuous.
	Applying Barbalat's Lemma, we deduce $\lim_{t \to \infty} \|\omega(t)\|^{2}=0$.
	Furthermore, by the same boundedness properties of the closed-loop system, the time derivative $\ddot{\omega}(t)$ is uniformly bounded wherever defined, which implies that $\dot{\omega}(t)$ is uniformly continuous. 
	Since $\lim_{t \to \infty} \omega(t) = 0$, applying Barbalat's Lemma again yields $\lim_{t \to \infty} \dot{\omega}(t) = 0$.
	Taking the limit as $t\to \infty$ in the closed-loop dynamics, we arrive at
	$\lim_{t \to \infty} \sk(R_{e}(t))^\vee =0$, which, combined with property (ii), implies property (iii).
\end{proof}

\subsection{Exponential Convergence}
\begin{thm}\label{Thm2}
Consider the closed-loop system with a constant desired attitude $R_d \in SO(3)$ and initial conditions $(\omega_{0},R_{0}) \in \mathbb{R}^{3} \times D$ satisfying $V(R_{0},\omega_{0},R_{0}) \leq \Gamma(R_{0})$.
Then there exist positive constants $c$ and $C$ such that the corresponding solution $(R,\omega,R_{g})$ satisfies
	\begin{align}
		\|\ro - \rmr{d}\| + \|\omega\| + \|\rmr{g} - \rmr{d}\|  \leq C e^{-c t}, \quad t>0.
	\end{align}
\end{thm}
\begin{proof}
	From properties (ii), (iii), and \eqref{boundG1}, we know that $V(R, \omega, R_g) \to 0$ and $\Gamma(\rmr{g}) \ge c_1 > 0$.
    Hence there exists $t_{\#} > 0$ such that for all $t \ge t_{\#}$, $\Gamma(R_{g})(t)-c_\Gamma V(R,\omega,R_{g})(t) \geq \frac{1}{2}c_{1}$.
	Using \eqref{cond1}, we observe that for $t \geq \max\{t_{\#},t_{*}\}$,
	\begin{align*}
		\frac{d}{dt} P(R_{g};R_{d})
		 &= (\Gamma-V) \frac{\kappa\tr({\rm grad}P(R_{g};R_{d})^\nlT \rho_{n}) }{\varepsilon} \\
		 &\leq  -\frac{c_0 c_{1}\kappa}{2\varepsilon}  P(R_{g};R_{d}).
	\end{align*}
	This immediately yields the desired decay estimate:
	\begin{align}
		\phi(\rmr[t]{g} \rmr{d}) + P_r = P(R_{g}(t);R_{d})  \leq C e^{-c t}, \quad t>0.
		\label{conver1}
	\end{align}

	To establish the exponential convergence of the tracking error, we define the modified Lyapunov function $V_{\alpha}(R,\omega,R_{g}):=\frac12 \omega^{\nlT} J \omega + k_{p}\phi(R_{e}) -\alpha \omega^{\nlT} \sk(R_{e})^\vee$, where $\alpha > 0$ is a small constant ensuring $V_\alpha$ is positive definite and equivalent to $V$.
	Using \eqref{eq:rot-sys} and \eqref{eq:taulaw}, the time derivative $\dot{V}_{\alpha}$ yields $\dot{V}_{\alpha} \le - (k_{d}/2) \|\omega\|^{2} - \alpha c_k \| \sk(R_{e})^\vee\|^{2} + C \|\dot{R}_{g}\|$ for some constants $c_k > 0$ and $C > 0$. 
	By applying the standard scalar inequality $ab \le \frac{\gamma}{2} a^2 + \frac{1}{2\gamma} b^2$ (valid for any $\gamma>0$) to the cross terms and choosing $\alpha$ appropriately, we obtain the required quadratic bounds.
	Since $\|\dot{R}_{g}\| \leq C e^{-{c_{1}} t}$ from \eqref{conver1}, there exists $t_2$ sufficiently large such that for $t \geq t_2$, $\dot{V}_{\alpha} + c_2 {V}_{\alpha} \leq C e^{-c_1 t}$ for some $c_{2} > 0$.
	By applying a standard differential inequality argument, we obtain $V(R(t),\omega(t),R_{g}(t)) \leq C e^{-c t}$ for some $c>0$.
	Finally, using the property $\|Q A\| = \|A\|$ for any orthogonal matrix $Q$, the total tracking error satisfies $\|\ro - \rmr{d}\| \leq \|\ro - \rmr{g}\| + \|\rmr{g} - \rmr{d}\| = 2 \sqrt{\phi(\rmr[t]{} \rmr{g})} + 2 \sqrt{\phi(\rmr[t]{g} \rmr{d})}$, completing the proof.
\end{proof}

\section{NUMERICAL SIMULATION}
\label{sec:sim}

\begin{figure}[t]
	\centering
	\includegraphics[width=0.75\linewidth]{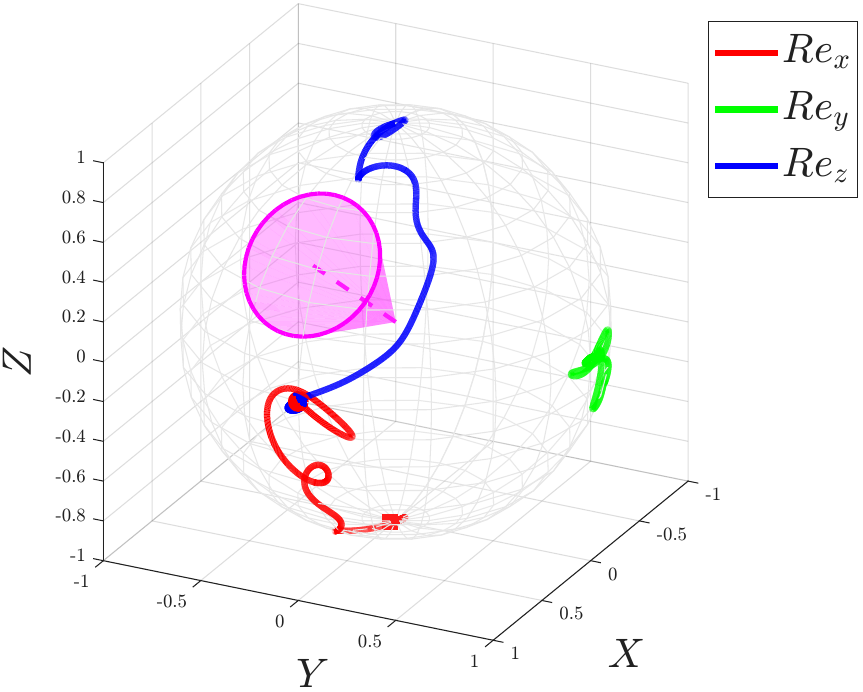}
	\caption{Evolution of system trajectories (PET-ERG). The body $z$-axis safely navigates around the forbidden conic region to reach the desired orientation.}
	\label{fig:atti_tra_pet}
\end{figure}
\begin{figure}[t]
	\centering
	\includegraphics[width=0.75\linewidth]{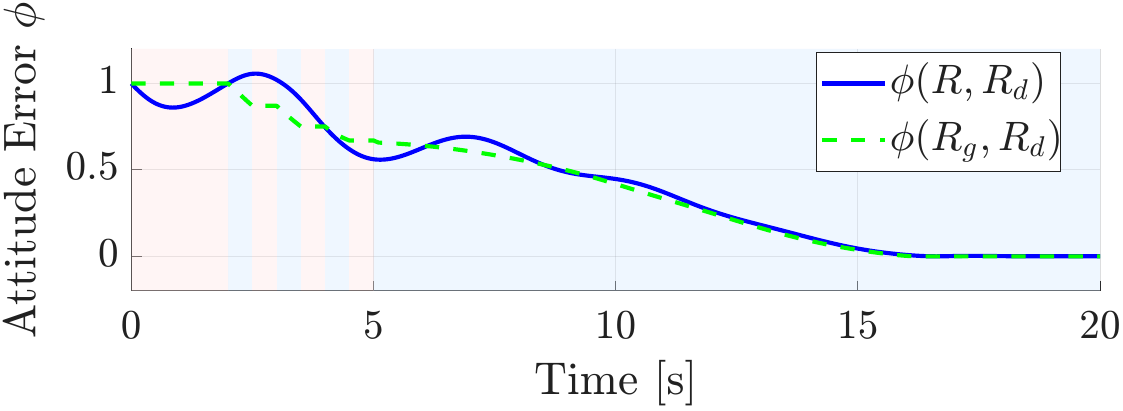}
	\caption{Evolution of attitude error $\phi(\rmr{}, \rmr{d})$ and reference error $\phi(\rmr{g}, \rmr{d})$ under the PET-ERG scheme.}
	\label{fig:attitude_error_pet}
\end{figure}
\begin{figure}[t]
	\centering
	\includegraphics[width=0.75\linewidth]{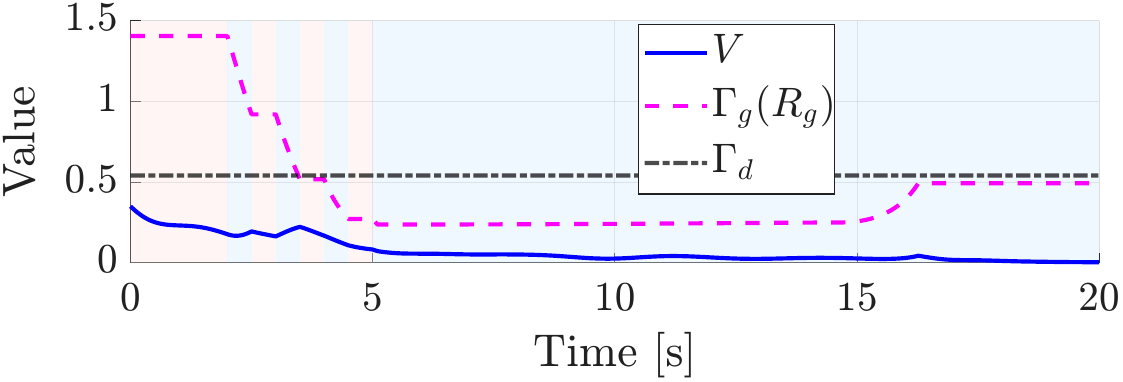}
	\caption{Evolution of the Lyapunov function $V$ and the safety thresholds $\Gamma_d$ and $\Gamma_g$. The periodic event-triggered mechanism maintains $V$ strictly below the aggregate margin.}
	\label{fig:gamma_pet}
\end{figure}
\begin{figure}[t]
	\centering
	\includegraphics[width=0.75\linewidth]{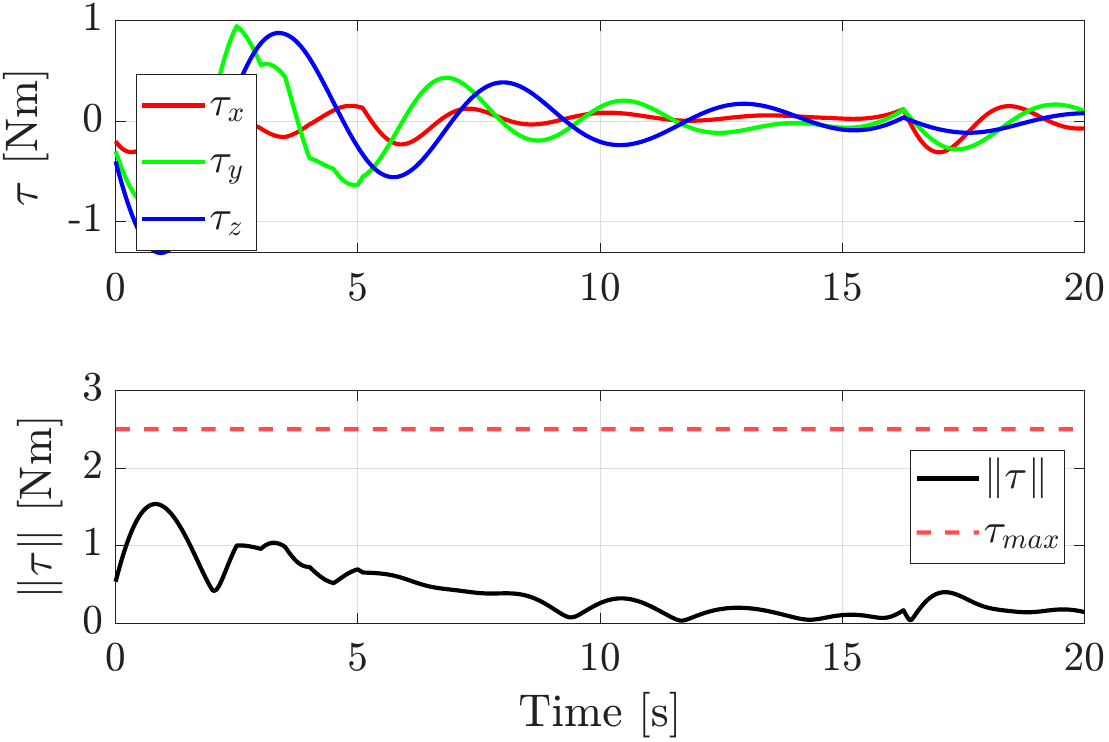}
	\caption{Evolution of the control torque input $\ve{\tau}$, strictly respecting the saturation limit $\tau_{\max}$.}
	\label{fig:tau_input_pet}
\end{figure}

In this section, we carry out numerical simulations to verify the effectiveness of the proposed PET-ERG scheme applied to the attitude control of a rigid body on $SO(3)$.
To evaluate the efficacy of the proposed method in handling non-convex constraints, we configure a reorientation scenario where the geodesic interpolation on $SO(3)$ connecting $\rmr{}(0)$ and $\rmr{d}$ intersects the forbidden cone, which can be verified by directly evaluating \eqref{eq:constraint_pointing} along the interpolation.
We consider an asymmetric rigid body with the inertia matrix $\ma{J} = \mathrm{diag}(1, 2, 3)~\si{\kilogram\cdot\meter\squared}$.
The control objective is to steer the attitude from the initial identity matrix $\rmr{}(0) = \ma{I}_3$ to the desired attitude $\rmr{d} = \exp(\frac{\pi}{2} \widehat{\ve{e}}_y)$, corresponding to a $90^\circ$ rotation around the body $y$-axis. The initial angular velocity is $\ve{\omega}(0) = [0.2~0.3~0.4]^\T~\si{\radian\per\second}$.

The system is subject to the input saturation \eqref{eq:constraint_input} and the geometric pointing constraint \eqref{eq:constraint_pointing}.
The maximum control torque is $\tau_{\max} = \SI{2.5}{\newton\cdot\meter}$.
The pointing constraint is defined by the body-fixed axis $\ve{a}_b = [0~0~1]^\T$, the inertial vector $\ve{a}_c = [-0.791~0.061~-0.609]^\T$, and the cone angle $\theta_c = 160^\circ$.
For the repulsive potential field, the parameters are set to $\delta = \cos(\theta_c) + 0.05$, $\zeta = \delta + 0.05$, $\varepsilon = 10^{-3}$, and $\varepsilon_\Gamma = 0.1$.

The controller gains in \eqref{eq:taulaw} are set to $k_p = 5$ and $k_d = 1$.
Given the input constraint $\tau_{\max}$ and the controller gains, the safe level set value $\Gamma_d$ is computed offline by solving the optimization problem described in \cite{nakano2023}.
The PET-ERG parameters are chosen as $\kappa = 1.0$, $c_\Gamma = 3$, and $T_s = \SI{0.5}{\second}$.
The initial reference state is set to $\rmr{g}(0) = \ma{I}_3$, aligning with the initial system attitude.

Figs. \ref{fig:atti_tra_pet}--\ref{fig:tau_input_pet} validate the theoretical properties established in Section \ref{sec:proposed_method}.
Fig. \ref{fig:atti_tra_pet} depicts the 3D evolution of the body axes, showing that the body $z$-axis successfully avoids the forbidden conic region (cyan) while converging to the desired orientation (magenta).
Fig. \ref{fig:attitude_error_pet} demonstrates the exponential convergence of both the actual attitude error $\phi(\rmr{}, \rmr{d})$ and the auxiliary reference error $\phi(\rmr{g}, \rmr{d})$ to zero, validating Theorem \ref{Thm2}.
As shown in Fig. \ref{fig:gamma_pet}, the Lyapunov function $V(t)$ is maintained below the aggregate dynamic threshold $\Gamma = \min\{\Gamma_d, \Gamma_g(\rmr{g}(t))\}$. The blue- and red-shaded backgrounds indicate intervals where the safety condition \eqref{eq:event_condition} is satisfied (reference updated) and violated (reference held), respectively. The periodic event-triggered logic intermittently suspends reference updates at each sampling instant $T_s$, introducing intervals of constant auxiliary reference that are consistent with the cascade stability analysis. During the continuous update phases, the robust margin parameter $c_\Gamma$ prevents any inter-sample constraint violations, which is consistent with the forward invariance result of Theorem \ref{Thm1}.
Finally, Fig. \ref{fig:tau_input_pet} confirms that the control torque satisfies the saturation limit $\|\ve{\tau}\| \leq \tau_{\max}$ at all times.

\section{CONCLUSIONS}
\label{sec:conclusion}
This letter proposed a Periodic Event-Triggered Explicit Reference Governor (PET-ERG) for constrained attitude control on $SO(3)$.
The proposed scheme ensures the satisfaction of input saturation and pointing constraints without relying on online optimization.
In contrast to continuous-time schemes, the proposed hybrid supervisory structure provides the structural properties necessary to rigorously establish the asymptotic stability and exponential convergence of the overall cascade system for almost all initial configurations.
Numerical simulations demonstrated the effectiveness of the proposed method in handling constraints while achieving precise attitude tracking.

\end{document}